%% file: main.tex
\documentclass[letterpaper, 10 pt, conference]{ieeeconf}  
\IEEEoverridecommandlockouts                              
                                                          
\usepackage{epsfig}
\usepackage{url}
\usepackage{amsmath}
\usepackage{amsfonts}

\usepackage{amsthm}
\usepackage{amssymb}
\usepackage{graphicx}
\usepackage{color}
\usepackage{balance} 
\usepackage{algorithm2e}
\usepackage{algorithmic}
\usepackage{bm}

\newtheorem{prop}{Proposition}

\let\labelindent\relax

\usepackage{enumitem}
\usepackage{graphicx} 
\graphicspath{{figure/}}
\usepackage{float}
\usepackage[caption=false,font=footnotesize]{subfig}
\usepackage{soul}
\usepackage{color}

\usepackage{algorithm,algorithmic}
\usepackage{cite}

\allowdisplaybreaks
\usepackage{mathtools}

\newtheoremstyle{bfnote}%
{}{}%
{\itshape}{}%
{\bfseries}{.}%
{ }%
{\thmname{#1}\thmnumber{ #2}\thmnote{ (#3)}}
\theoremstyle{bfnote}
\newtheorem{thm}{Theorem}
\newtheorem{asp}{Assumption}
\newtheorem{lem}{Lemma}
\newtheorem{rmk}{Remark}
\newtheorem*{base*}{Base Controller}

\usepackage{tikz}
\newcommand*\circled[1]{\tikz[baseline=(char.base)]{\node[shape=circle,draw,inner sep=0.05pt] (char) {#1};}}

\newcommand{\real}[0]{\mathbb R}

\DeclareSymbolFont{bbold}{U}{bbold}{m}{n}
\DeclareSymbolFontAlphabet{\mathbbold}{bbold}

\DeclarePairedDelimiterX\Set[2]{\lbrace}{\rbrace}%
{ #1 \,\delimsize| \,\mathopen{} #2 }

\usepackage{bookmark}

\title{\LARGE \bf Leveraging Predictions in Power System Frequency Control: an Adaptive Approach}

\author{Wenqi Cui, Guanya Shi, Yuanyuan Shi and Baosen Zhang
\thanks{Wenqi Cui and Baosen Zhang  are with the Department of Electrical and Computer Engineering, University of Washington Seattle, WA 98195, USA \{wenqicui, zhangbao\}@uw.edu} %
\thanks{Guanya Shi is with the  Paul G. Allen School of Computer Science and Engineering, University of Washington Seattle,  guanyas@cs.washington.edu }
\thanks{Yuanyuan Shi is with the Department of Electrical and Computer Engineering, University of California San Diego, yyshi@eng.ucsd.edu}
\thanks{The authors are partially supported by the NSF grants ECCS-1930605, 2200692, and 2153937.}}

\begin{document}
\maketitle
\thispagestyle{empty}
\pagestyle{empty}

\begin{abstract}
 Ensuring the frequency stability of electric grids with increasing renewable resources is a key problem in power system operations. In recent years, a number of advanced controllers have been designed to optimize frequency control. These controllers, however, almost always assume that the net load in the system remains constant over a sufficiently long time. Given the intermittent and uncertain nature of renewable resources, it is becoming important to explicitly consider net load that is time-varying.    

 This paper proposes an adaptive approach to frequency control in power systems with significant time-varying net load. We leverage the advances in short-term load forecasting, where the net load in the system can be accurately predicted using weather and other features. We integrate these predictions into the design of adaptive controllers, which can be seamlessly combined with most existing controllers including conventional droop control and emerging neural network-based controllers. We prove that the overall control architecture achieves frequency restoration decentralizedly. Case studies verify that the proposed method improves both transient and frequency-restoration performances compared to existing approaches.

\end{abstract}



\section{Introduction}
\input{introduction}

\section{Model and Problem Formulation} \label{sec:model}
\input{model}

\section{Leveraging predictions in power systems}
\label{sec:prediction}
\input{prediction}

\section{Modular Design of Adaptive Controllers } \label{sec:module}
\input{module}



\section{Case Study} \label{sec:simulation}
\input{simulation}

\section{Conclusion} \label{sec:conclusion} 
This paper proposes an adaptive approach that can be built on  conventional droop control or emerging neural network-based controllers to efficiently adapt to time-varying net load through predictions. We prove that frequencies will converge to their nominal value under the compositional adaptive controller design. Case studies on IEEE 39-bus and 145-bus systems demonstrate  that the proposed method  improves both transient and frequency-restoration performances compared with existing
approaches, while conventional droop control is not able to remove large oscillations and cannot always drive frequency to its nominal value.  Future directions include rigorous analysis of the performance guarantee subject to noises and prediction errors, and case studies using predictions from real system data.

\bibliographystyle{IEEEtran}
\bibliography{Reference}
\end{document}

%% file: introduction.tex

Frequency control plays a fundamental role in power system operations by maintaining the real-time supply and demand balance~\cite{sauer2017power}. With the increase in renewable resources, frequency stability has received significant interest because of the reduction in mechanical inertia and observed performance degradation~\cite{kroposki2017achieving}. Much of the recent works have focused on designing control algorithms that leverage the power electronic interfaces of renewable resources. Since devices like batteries and solar PVs are connected to the grid through inverters, they can quickly adjust their power to implement sophisticated control laws in response to frequency changes. For example,~\cite{cui2022tps, yuan2022reinforcement, jiang2022stable, cui2022structured, Zhao2015acc, SCHIFFER2017auto, weitenberg2018exponential} have been proposed to design optimal nonlinear controllers that outperform linear droop control. These include neural network-based controllers in~\cite{cui2022tps, yuan2022reinforcement, jiang2022stable, cui2022structured} to overcome the challenges in optimizing over functional space of stabilizing nonlinear controllers. In addition, steady-state frequency restoration can be realized distributedly through a method called distributed averaging-based integral (DAI)~\cite{Zhao2015acc, SCHIFFER2017auto, weitenberg2018exponential,jiang2022stable, cui2022structured}. 

A key assumption made in all of the above works is that there exists a separation of timescales. Namely, after a disturbance occurs--e.g., a change in net load--the system will not experience other disturbances until the frequency has recovered and reached a steady state. However, because of the intermittent and uncertain nature of renewable resources, the assumption that the net load will remain constant for a sufficiently long time until the system settles down may no longer be valid. Thus, it becomes critical to explicitly model the time-varying nature of net load in frequency control.

The changes in net load are typically modeled as noises in the system and there are various ways to deal with them. If the controllers for the noiseless system (i.e., the nominal system) are exponentially stable, then bounded noises will lead to bounded state deviations~\cite{khalil2002nonlinear}. Alternatively, robust controllers can be designed for the ``worst-case'' noise, again assuming the noise is bounded~\cite{bevrani2015robust,tabas2022computationally}. However, as the intermittency and uncertainties grow in the net load, both methods tend to lead to very conservative results. 

In this paper, we leverage the fact that although the net load is time-varying, it is largely predictable. In recent years, short-term net load  forecasting--spanning seconds to hours--has shown impressive performance~\cite{kusiak2010short,wang2017data , charytoniuk1998nonparametric, lu2021ultra, sun2019using}. Using features such as weather data and calendar information, the future net load can be predicted with a high degree of accuracy (see,~e.g.~\cite{kusiak2010short,wang2017data,  sun2019using, lu2021ultra, charytoniuk1998nonparametric} and the references within). Therefore, these predictions could be used in the control problem to mitigate the impact of time-varying net load. Previously, the features have been used in end-to-end learning methods that attempt to learn a mapping directly from feature observations 
to control actions~\cite{su2020adaptive, chen2022reinforcement}. However, it is challenging to provide performance guarantees from the learned controllers based on these  end-to-end learning approaches.

We bridge the gap between load predictions  and frequency control by designing controllers with provable performance guarantees under time-varying net load. Motivated by recent works on adaptive control~\cite{shi2021meta, o2022neural,huang2022robust}, we think of load forecasting as finding a set of basis functions representing the time-varying patterns of the net load, e.g., functions of wind speed, solar radiation, temperature, past trajectories of load, etc. Then we build an adaptive control law based on (possibly unknown) combination of basis functions.  

The adaptive control law  can be seamlessly integrated with conventional droop control or emerging neural network-based controllers~\cite{cui2022tps, yuan2022reinforcement, jiang2022stable, cui2022structured}, and we prove the frequencies will converge to their nominal value under such compositional controller design. Case studies on IEEE 39-bus and 145-bus power systems demonstrate that the adaptive control approach improves both transient and frequency-restoration performance compared to existing approaches.

%% file: model.tex
\subsection{Notations }
Throughout this manuscript, vectors are denoted in lower-case bold and matrices are denoted in upper-case bold, scalars are unbolded unless otherwise specified. Vectors of all ones and zeros are denoted as  $\mathbbold{1}_n, \mathbbold{0}_n \in \real^n$, respectively. For $\bm{\theta} \in \mathbb{R}^{ n}$, $\theta_i$  represent its $i$-th element. Superscript $^*$ indicates the equilibrium value of a variable. 

\subsection{Model }
Consider an $n$-bus power system modeled as a connected graph $\left(\mathcal{V},\mathcal{E} \right)$. Specifically, buses are indexed by $i,j \in \mathcal{V} := \left[n\right]:=\{1,\dots, n\} $ and transmission lines are denoted by unordered pairs $\{i,j\} \in \mathcal{E} \subset \Set*{\{i,j\}}{i,j\in\mathcal{V},i\not=j}$. 
We adopt the commonly used lossless power flow model and assume that all voltages are at $1$  per unit~\cite{sauer2017power}. Then the state variables are   phase angle $\boldsymbol{\theta}:=\left(\theta_i, i \in \left[n\right] \right) \in \real^n$ and frequency deviation from the nominal value $\boldsymbol{\omega}:=\left(\omega_i, i \in \left[n\right] \right) \in \real^n$. Denote $M_i$ and $D_i$ as the inertia constant and damping coefficient at bus $i$. The susceptance of the line $\{i,j\} \in \mathcal{E}$ is $B_{i j}=B_{ji }>0 $, and the susceptance $B_{kl}=B_{lk}=0$ when $\{k, l\} \notin \mathcal{E}$. 
Then, the frequency dynamics is given by the swing equation written as~\cite{sauer2017power}, 
    \begin{align}\label{eq:dyn}
\dot{\theta_i}(t) &=\omega_i(t)\,, \\
 M_{i}\dot{\omega}_{i}(t) &\!=\!p_{i}(t)\!-\! D_i \omega_i(t)\!-\!u_i(t)\!-\!\sum_{j=1}^{n}\!B_{i j}\sin(\theta_i(t)\!-\!\theta_j(t)) \nonumber,
    \end{align}
where $u_i(t)$ is the control action that changes the active power injection, $ p_{i}(t)$ is the net power injection (the negative of net load) at bus $i$, which is time-varying.  

We consider static local feedback controllers: bus $i$ measures its local frequency deviation $\omega_i$ and applies a static function to determine the control action $u_i$.  We envision the control actions coming from inverter-connected resources, which can follow almost arbitrary control law because of their fast-response capabilities.

\subsection{Optimal frequency control}

The goal of frequency control is to keep the frequency around their reference value (e.g., 60Hz in the United States).
Fig.~\ref{fig: freq_period} demonstrates typical frequency dynamics after a step increase in load, and we are interested in optimizing the transient process while realizing frequency restoration:



%

\begin{figure}[ht]
\centering
\includegraphics[width=0.9\columnwidth]{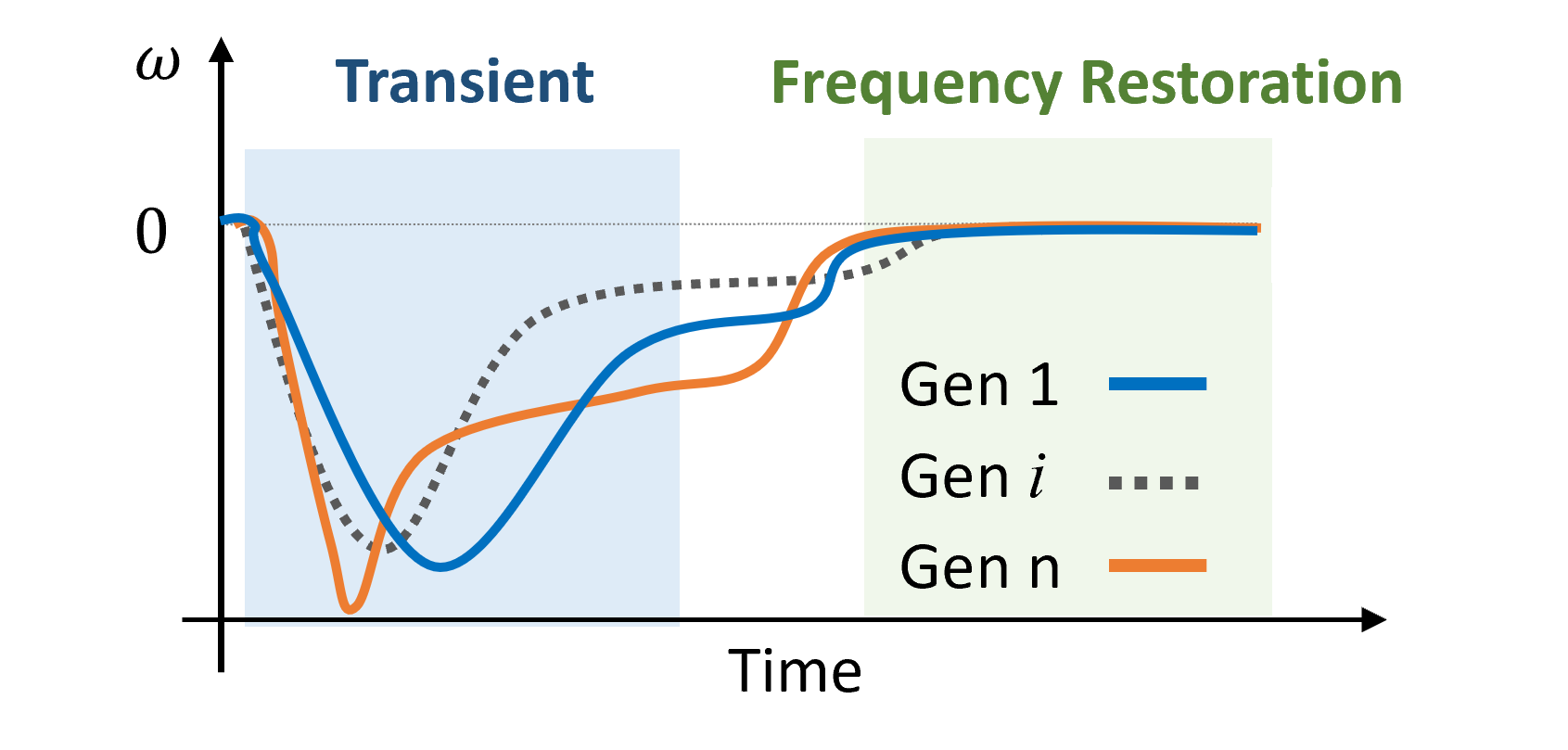}
\caption{ Frequency dynamics after a step increase of load. We are interested in optimizing the transient process while realizing frequency restoration. 
}
\label{fig: freq_period}
\end{figure}

\subsubsection{Frequency restoration} 
We seek to restore the frequency to their reference value after disturbances, i.e., $ \lim_{t\to\infty} \omega_i(t)=0\,, \forall i$. This goal is none trivial because of the time-varying nature of $p_i(t)$. We will later show that under some realistic models of $p_i(t)$, it is possible to realize frequency restoration using an adaptive control law.
\subsubsection{Transient period} 
This is the period after a disturbance and before the system settles down.  During the transient period, the most important goal is to reduce the maximum frequency deviation (also called the frequency nadir), represented by  the infinity norm of $\omega_i(t)$ over the time horizon from 0 to $T$, i.e., $\|\omega_i\|_{\infty}:=\sup_{0\leq t\leq T} |\omega_i(t)|$. We also want to avoid large control efforts, computed as $\int_{t=0}^T C_i(u_i(t))dt$ where $C_i(\cdot)$ is a Lipschitz-continuous cost function for $i\in[n]$.
The transient optimization problem up to time $T$ is~\footnote{The cost function can also include the summation of frequency deviation or other terms, and this does not change the analysis.}
\vspace{-0.1cm}
\begin{subequations}\label{eq: transient_optimization}
\begin{align}
    \min_{\bm{u}(\cdot)}& \sum_{i=1}^n \left(\|\omega_i\|_{\infty}+\gamma \int_{t=0}^TC_i(u_i(t))dt\right)\,, \label{subeq: transient_cost}\\
    \text{s.t. } & 
    \text{dynamics in } \eqref{eq:dyn},
     \label{subeq: opt_dyn}\\
    &\lim_{t\to\infty} \omega_i(t) =0, \forall i\in[n] \label{subeq: opt_stable}
    \end{align}
\end{subequations}
where $\gamma$ is a coefficient that trades off the cost of action
with respect to the frequency deviation. The  requirement of frequency restoration is represented as a constraint in~\eqref{subeq: opt_stable}. 



Even when $p_i(t)$ is time-invariant, the nonlinear dynamics~\eqref{subeq: opt_dyn}, the frequency restoration requirement~\eqref{subeq: opt_stable} and the nonconvex cost functions make it challenging to solve~\eqref{eq: transient_optimization} using conventional optimization techniques. Recent works~\cite{cui2022tps, yuan2022reinforcement, cui2022structured, jiang2022stable} propose to learn functions $\bm{u}(\cdot)$ by parameterizing them as neural networks and train them by minimizing the cost in~\eqref{subeq: transient_cost}. These works derive structural properties of stabilizing controllers, and enforce the structures in neural networks to guarantee~\eqref{subeq: opt_stable} by design. 
However, these analyses rely on the assumption that $p_i(t)$ is time-invariant. The rest of the paper considers what happens when $p_i(t)$ is time-varying.

%% file: prediction.tex
This section explicitly represents time-varying net power injections for frequency control. Then, we show the intuitions for developing adaptive  law in controllers. 
\subsection{Modeling time-varying power injections}
We model the time-varying net power injection in~\eqref{eq:dyn}  as 
\begin{equation}\label{eq: p_t}
    p_i(t) = p_i^*+\bm{\phi}_i(t)^\top \bm{a}_i\, ,
\end{equation}
where $p_i^*$ is the setpoint  of net power injection (e.g., the economic dispatch solution) that satisfies DC power flow $p_i^* = \sum_{j=1}^{n}B_{i j}\sin(\theta_i^*\!-\!\theta_j^*) \,\forall i\in[n]$. The second term, $\bm{\phi}_i(t)^\top \bm{a}_i$, captures the variation in time. More specifically, we think $\bm{\phi}_i(t)$ as a vector of basis functions of features and $\bm{a}_i$ as some coefficients (possibly unknown). Note that both $\bm{\phi}$ and $\bm{a}$ are indexed by $i$ and different buses may have a different set of basis functions and coefficients. 


The model in \eqref{eq: p_t} comes from the fact that the net power injections typically depend on a set of time-varying features (or basis functions), including weather, temperature, time of the day, solar irradiation, wind conditions, historical load and others~\cite{kusiak2010short,wang2017data,  sun2019using, lu2021ultra, charytoniuk1998nonparametric}. If $\bm{\phi}_i(t)$ includes these features, then $\bm{\phi}_i(t)^\top \bm{a}_i$ is simply a nonlinear time series-based load forecasting algorithm~\cite{lu2021ultra, charytoniuk1998nonparametric, kusiak2010short,wang2017data, sun2019using}. In addition, $\bm{\phi}_i(t)$ can include different kernel functions that have been used in accurate forecasting algorithms~\cite{shi2021meta, o2022neural,huang2022robust}. Note that we always include the scalar $1$ as part of the feature kernels to capture the base load.  
The coefficients, $\bm{a}$, are the weights assigned to each of the features that make up the variations in the net load. Our approach does not assume that the value of $\bm{a}$ is known, therefore as long as the functional form of \eqref{eq: p_t} is correct, we can design a controller that will adapt to the time-varying deviations. In practice, this means that we can include a large set of features to help fully capture the behavior of the net load. 


\subsection{ Intuition on the adaptive control law}
Since the frequency dynamics of the system in~\eqref{eq:dyn} depends
only on the phase angle differences between different nodes,  we make the change of
coordinates $\delta_i(t):=\theta_i(t)-\frac{1}{n}\sum_{j=1}^{n} \theta_j(t)$ and write the  frequency dynamics in~\eqref{eq:dyn} in the center-of-inertia coordinates, which are more convenient for analysis~\cite{sauer2017power,weitenberg2018robust}:

\vspace{-0.2cm}
\small   \begin{align}\label{eq:dyn_change_cor}
\dot{\delta_i}(t) &=\omega_i(t)-\frac{1}{n}\sum_{j=1}^n\omega_j(t)\,, \\
M_{i}\dot{\omega}_{i}(t) &\!=p_i^*\!+\!\bm{\phi}_i(t)^\top \!\bm{a}_i\!-\! D_i \omega_i(t)\!-\!u_i(t)\!-\!\sum_{j=1}^{n}\!B_{i j}\sin(\delta_{ij}(t)) \nonumber,
    \end{align}
\normalsize
where $\delta_{ij}(t):=\delta_i(t)-\delta_j(t)$ and we plug in the expression of $p_i(t)$ in~\eqref{eq: p_t}.  

Intuitively,  if the controller has the form $u_i(t) = \hat{u}_i(\omega_i(t))+\bm{\phi}_i(t)^\top\bm{a}_i$, it will cancel the impact of net load variations in $\bm{\phi}_i(t)^\top\bm{a}_i$.
However, this is not possible since $\bm{a}_i$ is not known. Inspired by works on adaptive control~\cite{slotine1991applied,o2022neural,shi2021meta}, we design an adaptation law $\bm{\phi}_i(t)^\top\bm{\hat{a}}_i(t)$  where $\dot{\bm{\hat{a}}}_i(t)$ is updated according to observations in the next section.

%% file: module.tex
In this section, we derive the modular controller design that consists of a base control module and an adaptation law. Then we provide conditions on the control law that guarantees frequency restoration. 



\subsection{Modular adaptive control law}
 
We adopt a modular approach for the controller design: the control law $\hat{u}_i(\omega_i(t))$ for the system with time-invariant net load, and the adaptation law $\bm{\phi}_i(t)^\top\bm{\hat{a}}_i$ 
 for the variation of the net load in time.   Compactly,  the control law is
\begin{subequations}\label{eq:control_adpt}
\begin{align}
& u_i(t) = \underbrace{\hat{u}_i(\omega_i(t))}_{\text{Base controller}}+\;\underbrace{\bm{\phi}_i(t)^\top\bm{\hat{a}}_i(t)}_{\text{Adaptive Law}}\\
&\dot{\bm{\hat{a}}}_i(t)=\omega_i(t)\cdot\bm{A}_i\bm{\phi}_i(t) ,\label{subeq: dot_hat_a}
\end{align}    
\end{subequations}
where $\bm{A}_i\succ 0$ is a tunable matrix and we will specify it later in Section~\ref{subsec: training}.

To specify the base controller, we make the following assumption for the ranges of the angle difference. 
\begin{asp}\label{asp: angle_diff}
$\forall \{i,j\} \in\mathcal{E}$, the bus voltage angle difference is within $\pm \pi/2$, 
i.e., $\forall \{i,j\} \in\mathcal{E}$, $|\delta_i(t)-\delta_j(t)|\in [0,\pi/2)$.
\end{asp}
Since the operation of power systems normally restricts the angle differences to be much smaller than $\pi/2$~\cite{sauer2017power}, this assumption holds in most scenarios. Under Assumption~\ref{asp: angle_diff}, there is a general class of controllers that guarantees the local exponential stability of the system if $p_i(t)$ is time-invariant. We call these controllers as the Base Controller:
\begin{base*}\label{base: control}
Any controller $\hat{u}_i(\omega_i(t)) \,\forall i\in[n]$ that is monotonically increasing and across the origin is  locally exponential stable for the system if $p_i(t)$ is time-invariant~\cite{cui2022tps}. For example, this class of controllers includes 
\begin{enumerate}[label=(\roman*)]
    \item Droop control $\hat{u}_i(\omega_i(t))=\varphi_i\omega_i(t)$ with $\varphi_i>0$. This is the most popular form of controllers implemented in primary frequency control.  
    \item Neural network-based control law $\hat{u}_i(\omega_i(t))=\hat{u}_{\bm{\psi}_i}\left(\omega_i(t)\right)$, where $\hat{u}_{\bm{\psi}_i}(\cdot)$ is the function parameterized by monotone neural networks with $\bm{\psi}_i$ being trainable parameters in~\cite{cui2022tps}. These monotone neural network-based controllers can offer better performances compared to linear ones~\cite{cui2022tps, jiang2022stable}. 
\end{enumerate}
\end{base*}
\begin{rmk}\label{rmk: PLlinear}
If $\phi_i(t)=1$, then~\eqref{subeq: dot_hat_a} reduces to a generic integral control where $\dot{\hat{a}}_i(t)=A_i\omega_i(t)$ and the compositional controller~\eqref{eq:control_adpt} reduces to a proportional-integral (PI) controller~\cite{cui2022structured}. Hence, the adaptation law in~\eqref{subeq: dot_hat_a} can be viewed as an extended version of integral control that makes use of the basis functions in $\bm{\phi}_i(t)$. In the experiment, we will show that the control law in~\eqref{eq:control_adpt} can reduce the frequency oscillations compared to a standard PI controller. 
\end{rmk}

\subsection{Frequency restoration}




 
 The next theorem states conditions on the system with adaptive control law in~\eqref{eq:control_adpt} to guarantee frequency restoration.
 
\begin{thm}\label{thm: convergence}
Suppose Assumption~\ref{asp: angle_diff} holds and the control law is~\eqref{eq:control_adpt} where $\hat{u}_i(\omega_i(t))$ is a monotonically increasing function cross the origin. For the closed-loop system formed by~\eqref{eq:dyn_change_cor} and ~\eqref{eq:control_adpt}, there exists a non-empty set $\mathcal{Q}_\rho$ such that if the initial states satisfying
$\left(\bm{\delta}(0), \bm{\omega}(0),\bm{\hat{a}}(0)\right)\in \mathcal{Q}_\rho $,  the frequency deviation $\bm{\omega}$ of each bus  converges to zero, i.e., $lim_{t\to\infty}\omega_i(t)=0 \, \forall i\in[n]$.
\end{thm}

This theorem 
covers a special case where $\phi_i(t)=1$, and the results align with~\cite{cui2022structured} that shows the PI controller guarantees frequency restoration. Here, we want to show that if the time-varying net load has a richer set of features $\bm{\phi}$, the adaptive control law can still maintain frequency restoration. The region of attraction $\mathcal{Q}_\rho$ is characterized in Proposition~\ref{prop: Invariance-like Theorem}. The subscript $\rho$ describes the size of this region, and we will show how it can be computed. For convenience, we omit $(t)$ for the time-varying variables in the rest of the paper. 

 
The main tool to prove Theorem~\ref{thm: convergence}  is the Invariance-like Theorem~\cite[Theorem~8.4]{khalil2002nonlinear}. We define $\bm{x}:=\left(\bm{\delta}, \bm{\omega},\bm{\hat{a}}\right)\in\real^{2n+\sum_{i=1}^nl_i}$ and write the closed-loop system formed by~\eqref{eq:dyn_change_cor} and~\eqref{eq:control_adpt} as $\dot{\bm{x}}=\bm{f}(t,\bm{x})$. Define the set for the states satisfying Asusmption~\ref{asp: angle_diff} as $\mathcal{X}:=\left\{\left(\bm{\delta}, \bm{\omega},\bm{\hat{a}}\right)| |\delta_i-\delta_j|\in [0,\pi/2)\,\forall \{i,j\} \in\mathcal{E}\right\}$, which contains $\bm{x}^*=(\bm{\delta}^*, \mathbbold{0}_n, \bm{a})$. The Invariance-like Theorem (in the notation of this paper) is:
\begin{prop}[Invariance-like
Theorem~\cite{khalil2002nonlinear}] \label{prop: Invariance-like Theorem}
Consider the system $\dot{\bm{x}}=\bm{f}(t,\bm{x})$ where $\bm{x}\in \real^{2n+\sum_{i=1}^nl_i}$. Let $V:\mathcal{X} \rightarrow R$ be a continuously differentiable functions such that
\begin{align}
W_1(\bm{x}) &\leq V( \bm{x}) \leq W_2(\bm{x}) \\
\dot{V}(\bm{x})& \leq-W_0(\bm{x})    
\end{align}
hold on $\mathcal{X}$, where $W_1(\bm{x})$ and $W_2(\bm{x})$ are continuous positive definite functions and $W_0(\bm{x})$ is a continuous positive semidefinite function on $\mathcal{X}$. Choose $r>0$ such that $\mathcal{B}_r:=\left\{\bm{x}\in \real^{2n+\sum_{i=1}^nl_i}| \|\bm{x}\|<r\right\}\subset \mathcal{X}$ and let $\rho<\min _{\|x\|=r} W_1(\bm{x})$. Then, all solutions of the system with an initial condition $\bm{x}\left(0\right) \in\mathcal{Q}_\rho :=\left\{x \in \mathcal{B}_r \mid W_2(\bm{x}) \leq \rho\right\}$ are bounded and satisfy
$$
\lim _{t \rightarrow \infty} W_0(\bm{x}(t)) \rightarrow 0\, .
$$
\end{prop}
To use the above result, we construct a function $V(\bm{\delta}, \bm{\omega},\bm{\hat{a}})$ that  is bounded by continuous positive definite functions and  satisfying {$\dot{V}(\bm{\delta}, \bm{\omega},\bm{\hat{a}})\leq -\sum_{i=1}^{n} D_{i}\omega_i^{2}$. Then the convergence of $-\sum_{i=1}^{n} D_{i}\omega_i^{2}$ to zero follows directly from Proposition~\ref{prop: Invariance-like Theorem}. The rest of this section outlines the proof of Theorem~\ref{thm: convergence}.

We use the following energy function
\begin{align}\label{eq: Lyapunov_adaptive}
V(\bm{\delta}, \bm{\omega},\bm{\hat{a}})=\hat{V}(\bm{\delta}, \bm{\omega})\!+\!\frac{1}{2}\sum_{i=1}^n \left( \bm{\hat{a}}_i\!-\!\bm{a}_i\right)^\top\! \bm{A}_i^{-1}\!\left( \bm{\hat{a}}_i\!-\!\bm{a}_i\right)
\end{align}
where
$\hat{V}(\bm{\delta}, \bm{\omega})=: \frac{1}{2}\sum_{i=1}^{n} M_i\omega_{i}^{2}+W_\mathrm{p}(\bm{\delta})$
with
\begin{equation}
\begin{aligned}
    W_\mathrm{p}(\bm{\delta}) := 
&-\frac{1}{2}\sum_{i=1}^{n} \sum_{j=1}^{n} B_{ij}\left( \cos(\delta_{ij}) -\cos(\delta_{ij}^*)\right)\nonumber\\
&-\sum_{i=1}^{n}\sum_{j=1}^{n} B_{ij} \sin(\delta_{ij}^*)( \delta_{i}-\delta_{i}^*)\,.\label{eq:Wp}
\end{aligned}
\end{equation}






The next Lemma shows how $V$ can be bounded:
\begin{lem}[Bounds of the energy function]\label{lem: PD_Lyap}
For all $(\bm{\delta}, \bm{\omega},\bm{\hat{a}})\in\mathcal{X}$ , the energy function $V(\bm{\delta}, \bm{\omega},\bm{\hat{a}})$ in~\eqref{eq: Lyapunov_adaptive} satisfies 
$$W_1(\bm{\delta}, \bm{\omega},\bm{\hat{a}}) \leq V( \bm{\delta}, \bm{\omega},\bm{\hat{a}}) \leq W_2(\bm{\delta}, \bm{\omega},\bm{\hat{a}}) $$
with 
$$
\begin{aligned}
&W_1(\bm{\delta}, \bm{\omega},\bm{\hat{a}}) :=\gamma_{1}(\|\bm{\delta}-\bm{\delta}^*\|_2^{2}+\|\bm{\omega}\|_2^{2}+\| \bm{\hat{a}}\!-\!\bm{a}\|_2^2)\,,\\ 
& W_2(\bm{\delta}, \bm{\omega},\bm{\hat{a}}) := \gamma_{2}(\|\bm{\delta}-\bm{\delta}^*\|_2^{2}+\|\bm{\omega}\|_2^{2}+\| \bm{\hat{a}}\!-\!\bm{a}\|_2^2)\,,
\end{aligned}
$$
for some constants $\gamma_{1}>0$ and $\gamma_{2}>0$.
\end{lem}
\begin{proof}
We start by bounding the term $W_\mathrm{p}(\bm{\delta}) $. Define  $S(\bm{\delta}):=-\frac{1}{2}\sum_{i=1}^{N} \sum_{j=1}^{N} B_{ij} \cos(\delta_{ij}) $, then $\nabla S(\bm{\delta}) =  \left([\nabla S(\bm{\delta})]_i:=\sum_{j=1}^{n} B_{ij} \sin(\delta_{ij}), i \in \left[n\right] \right) \in \real^n$ and $\nabla^2 S(\bm{\delta}):=\left(\left[ \nabla^2 S(\bm{\delta})\right]_{i,j}=\frac{\partial^2 S(\bm{\delta})}{\partial \delta_i \delta_j}, i,j\in[n]\right)\in \real^{n \times n}$ where
\begin{align*}\label{eq:Lij}
     \frac{\partial^2 S(\bm{\delta})}{\partial \delta_i \delta_j}=
    \begin{dcases}
    -B_{ij}\cos(\delta_{ij}) & \text{if}\ i\neq j\\
    \sum_{j^\prime=1,j^\prime\neq i}^n B_{ij\prime}\cos(\delta_{ij\prime})& \text{if}\ i=j
    \end{dcases}\,,\forall i,j \in \left[n\right]\,.
\end{align*}
For $(\bm{\delta}, \bm{\omega},\bm{\hat{a}})\in\mathcal{X}$, $\nabla^2 S(\bm{\delta})$ is the Laplacian matrix that satisfying  $\nabla^2 S(\bm{\delta})\succ 0$  and thus $S(\bm{\delta})$ is convex. Note that
\begin{equation}\label{eq:Wp_bregman}
W_\mathrm{p}(\bm{\delta})=S(\bm{\delta})-S(\bm{\delta}^*)-\mathbf\nabla S(\bm{\delta}^*)^\top\left(\bm{\delta}-\bm{\delta}^{*}\right)    
\end{equation}
is the Bregman distance of the convex function $S(\bm{\delta})$ at the point $\bm{\delta}^*$.  According to~\cite[Lemma~4]{weitenberg2018exponential}, ~\eqref{eq:Wp_bregman}  can be bounded by $\beta_{1}\left\|\bm{\delta}-\bm{\delta}^{\ast}\right\|_2^2 \leq W_\mathrm{p}(\bm{\delta}) \leq \beta_{2}\left\|\bm{\delta}-\bm{\delta}^{\ast}\right\|_2^2 $  for $\bm{\delta}$  satisfying Assumption~\ref{asp: angle_diff} and some constants $\beta_{1},\beta_{2}>0$.

Using the Rayleigh-Ritz theorem~\cite{Horn2012MA}, the term  $\frac{1}{2}\sum_{i=1}^{n} M_i\omega_{i}^{2}$ is lower bounded by $\frac{1}{2}\left(\min_i M_i\right) \|\bm{\omega}\|_2^{2}$ and is upper bounded by $\frac{1}{2}\left(\max_i M_i\right) \|\bm{\omega}\|_2^{2}$. The term $\frac{1}{2}\sum_{i=1}^n \left( \bm{\hat{a}}_i-\bm{a}_i\right)^\top \bm{A}_i^{-1}\left( \bm{\hat{a}}_i-\bm{a}_i\right)$ is lower bounded by $\frac{1}{2}\left(\min_i \lambda_{min}(\bm{A}_i) \right)\| \bm{\hat{a}}\!-\!\bm{a}\|_2^2$ and upper bounded by $\frac{1}{2}\left(\max_i \lambda_{max}(\bm{A}_i) \right)\| \bm{\hat{a}}\!-\!\bm{a}\|_2^2$. Therefore, we can bound the energy function in~\eqref{eq: Lyapunov_adaptive} with
 \begin{align*}
\gamma_{1} &:=\dfrac{1}{2}\min \left(\min_i M_i, \, 2\beta_{1},\,\min_i \lambda_{min}(\bm{A}_i)  \right)>0\,, \\
\gamma_{2} &:=\dfrac{1}{2}\max \left(\max_i M_i,\, 2\beta_{2},\,\max_i \lambda_{max}(\bm{A}_i) \right)>0\,.
\end{align*}
\end{proof}

The next lemma shows an important inequality associated with the time derivative of the energy function and the positive definite function $W_0(\bm{\delta}, \bm{\omega},\bm{\hat{a}}) :=\sum_{i=1}^{n} D_{i}\omega_i^{2}$.

\begin{lem}[Time derivative of the energy function]\label{lem: dot_V}
The time derivative of the energy function satify 
$$ \dot{V}(\bm{\delta}, \bm{\omega},\bm{\hat{a}})\leq -W_0(\bm{\delta}, \bm{\omega},\bm{\hat{a}}) :=-\sum_{i=1}^{n} D_{i}\omega_i^{2}.
 $$ 
\end{lem}
\begin{proof}
We start by computing the partial derivatives of $V(\bm{\delta}, \bm{\omega},\bm{\hat{a}})$ with respect to each state. From ~\eqref{eq: Lyapunov_adaptive} and~\eqref{eq:Wp_bregman} , we have
\begin{equation}\label{eq: Lyapunov_derivative_partial}
\begin{split}
   \nabla_{\bm{\delta}}V &=\nabla S(\bm{\delta})-\nabla S(\bm{\delta}^*) \, ,\\
    \nabla_{\bm{\omega}}V &=\bm{M}\bm{\omega}\, ,\\
    \nabla_{\bm{\hat{a}}_i}V &=\bm{A}_i^{-1}\left( \bm{\hat{a}}_i-\bm{a}_i\right) \,\forall i\in[n].
\end{split}    
\end{equation}


The time derivative of $V(\bm{\delta}, \bm{\omega})$ defined in~\eqref{eq: Lyapunov_adaptive} is given by 
\begin{align}
&\dot{V}(\bm{\delta}, \bm{\omega},\bm{\hat{a}}) \label{eq: Vdot_convergence} \nonumber\\
=& \left(\nabla_{\bm{\delta}}V\right)^\top \dot{\bm{\delta}}+\left(\nabla_{\bm{\omega}}V\right)^\top \dot{\bm{\omega}}+\sum_{i=1}^n\left(\nabla_{\bm{\hat{a}}_i}V \right)^\top \dot{\bm{\hat{a}}}_i \nonumber
\\
\stackrel{\circled{1}}{=}& \left(\nabla S(\bm{\delta})-\nabla S(\bm{\delta}^*)\right)^\top\left(\bm{\omega}-\mathbbold{1}_n\frac{\mathbbold{1}_n^\top\bm{\omega}}{n}\right)\\
& \!+\!\sum_{i=1}^n \omega_i\left(p_i^*-\bm{\phi}_i(t)^\top\Tilde{\bm{a}}_i\!-\!D_i\omega_i\!-\!\hat{u}_i(\omega_i)\!-\![\nabla S(\bm{\delta})]_i\right) \nonumber\\
& - \!\sum_{i=1}^n  \omega_i\underbrace{\left(p_i^*-[\nabla S(\bm{\delta}^*)]_i\right)}_{=\bm{0}}
+\sum_{i=1}^n\omega_i\cdot\bm{\Tilde{a}}_i^\top \bm{A}_i^{-1}\bm{A}_i\bm{\phi}_i(t) \nonumber\\
=&\!-\!\sum_{i=1}^{n}\! D_{i}\omega_i^{2}\!-\!\sum_{i=1}^{n}\!\omega_i\hat{u}_i(\omega_i)\!-\!\left(\nabla \!S(\bm{\delta})\!-\!\nabla \!S(\bm{\delta}^*\!)\right)^\top\!\!\left(\mathbbold{1}_n\!\frac{\mathbbold{1}_n^\top\bm{\omega}}{n}\right) \nonumber\\
\stackrel{\circled{2}}{=}&  -\sum_{i=1}^{n}\! D_{i}\omega_i^{2}\!-\!\sum_{i=1}^{n}\!\omega_i\hat{u}_i(\omega_i) \nonumber \\
\stackrel{\circled{3}}{\leq} &-\sum_{i=1}^{n} D_{i}\omega_i^{2}, \nonumber
\end{align}
where $\bm{\Tilde{a}}_i:=\bm{\hat{a}}_i-\bm{a}_i$. The equality $\circled{1}$ plugs in~\eqref{eq: Lyapunov_derivative_partial} and adds the extra term using $p_i^* = \sum_{j=1}^{n}B_{i j}\sin(\delta_i^*\!-\!\delta_j^*) \equiv [\nabla S(\bm{\delta}^*)]_i \,\forall i\in[n]$.  Since $B_{ij}\sin(\delta_i-\delta_j)=-B_{ji}\sin(\delta_j-\delta_i)\, \forall i,j$, we have $\sum_{i=1}^n \sum_{j=1}^{n}B_{i j}\sin(\delta_i\!-\!\delta_j)=0$. Then the equality $\circled{2}$ follows from the property that $\mathbbold{1}_n^\top \nabla S(\bm{\delta}) = \sum_{i=1}^n \sum_{j=1}^{n}B_{i j}\sin(\delta_i\!-\!\delta_j)=0 $ for all $\bm{\delta}\in\real^n$. Because $\hat{u}_i(\omega_i)$ is monotonically increasing with $\omega_i$, we have $\hat{u}_i(\omega_i)\omega_i\geq 0\, \forall i $ and thus the inequality $\circled{3}$ holds.
\end{proof}

 By Proposition~\ref{prop: Invariance-like Theorem}, the bound of the energy function in Lemma~\ref{lem: PD_Lyap} and  $\dot{V}_1(\bm{\delta}, \bm{\omega},\bm{\hat{a}})\leq -\sum_{i=1}^{n} D_{i}\omega_i^{2}$ in Lemma~\ref{lem: dot_V} imply that $\sum_{i=1}^{n} D_{i}\omega_i^{2}\to 0$ as $t\to\infty$ for initial states satisfying $\left(\bm{\delta}(0), \bm{\omega}(0),\bm{\hat{a}}(0)\right)\in \mathcal{Q}_\rho $. Since $\sum_{i=1}^{n} D_{i}\omega_i^{2}= 0$ if and only if $\omega_i=0\, \forall i\in[n]$, we complete the proof of Theorem~\ref{thm: convergence}.

\subsection{Training}\label{subsec: training}
The matrix $\bm{A}_i\succ 0$ and the base control law $\hat{u}_i(\omega_i)$ should be chosen to optimize the behavior of the system. For $\bm{A}_i$'s, similar to how LQR controllers are tuned, we parameterize them as diagonal matrices and optimize these entries. For $\hat{u}_i(\cdot)$,  we adopt the learning-based method to parameterize $\hat{u}_i(\cdot)$ as a monotone neural network.  Note that this construction of $\hat{u}_i(\cdot)$ and $\bm{A}_i$  satisfies the conditions in  Theorem~\ref{thm: convergence},  then the frequency restoration always holds  by design. We optimize  $\hat{u}_i(\cdot)$ and $\bm{A}_i$ by training it with the loss function the same as  the optimization objective~\eqref{subeq: transient_cost}. Detailed formulation of the monotone neural network and the training framework can be found in our previous work~\cite{cui2022tps}.

%% file: simulation.tex
Case studies are conducted on the IEEE New England 10-machine 39-bus (NE39) and the IEEE 50-machine 145-bus test systems~\cite{athay1979practical,chow1992toolbox}.  We use TensorFlow 2.0 framework  in Google Colab with a single Nvidia Tesla T4 GPU with 16GB memory to train $\bm{A}_i$ and $\hat{u}_i(\cdot)$ for all $i\in[n]$.

\subsection{IEEE 39-bus test system.}
We conduct experiments on the IEEE-NE39 power network with parameters given in~\cite{athay1979practical, cui2022tps}.  For training and testing the controller, we generate  300 trajectories by randomly picking at most three buses to experience a step load change uniformly distributed in $\text{uniform}[-1,1]\,\text{p.u.}$, where 1p.u.=100 MW is the base unit of power. 
We discretize time and use $k$ to index the time step, with an interval of $0.01$s between each step.  The number of time steps in a trajectory in the training set is $K=400$, i.e., each trajectory spans the time horizon of $4$s. 
 For simplicity and illustrative purposes, we use sinusoidal signals to represent time-varying oscillations of the net load and demonstrate the performance of the proposed method. More advanced load forecasting algorithms (e.g., see ~\cite{charytoniuk1998nonparametric, wang2017data,lu2021ultra}
 and references within) can also be used without changes to the algorithm. The basis $\phi_i(k) = (sin(\eta_i^1 k),  sin(\eta_i^2 k), 1)$, where $\eta_i^1,\eta_i^2\sim \text{uniform}[0.005\pi, 0.02\pi]$. The coefficients $[\bm{a}_{i}]_j\sim \text{uniform}[0.1,0.2]$ for each item $j$ in $\bm{a}$. The training epoch number is 600. The loss function is~\eqref{subeq: transient_cost} where $\gamma=0.1$, $C_i(u_i) = c_i (u_i)^2$ with $c_i\sim \text{uniform}[0.025,0.075]$.

We compare the performance of 1) NN-Adaptive: the proposed adaptive controller~\eqref{eq:control_adpt} where $\hat{u}_i(\cdot)$ is parameterized by the monotone neural network with 20 neurons in the hidden layer, 2) NN-Integral: a standard integral controller (see Remark~\ref{rmk: PLlinear}) is used without considering the other features in $\bm{\phi}$, and $\hat{u}_i(\cdot)$ is the same as NN-Adaptive - parameterized by the monotone neural network with 20 neurons in the hidden layer; 3) Linear Droop: Conventional linear droop control with $\hat{u}_i(\omega_i) = \varphi_i\omega_i$ for all $i\in[n]$. The values of $\varphi_i$ are optimized through learning. 

The average batch loss  during epochs of training is shown in Fig.~\ref{fig:loss_n10}. All three methods converge, with the linear droop control having the highest loss. 
Fig.~\ref{fig:cost_n10} compares the transient and frequency-restoration cost on the test set. Specifically, the transient cost is~\eqref{subeq: transient_cost} for $T=4$s.  The frequency-restoration cost is the average frequency deviation in 10s-15s after the step change, where the dynamics approximately enter the steady state after $t=10$s as we will show later in the simulation. The transient cost of NN-Adaptive is 1.5\% lower than  NN-Integral and 67\% lower than Linear Droop, respectively.  The frequency-restoration cost of NN-Adaptive is 60\% lower than  NN-Integral and 94\% lower than Linear Droop, respectively. 
Hence, linear droop control has significantly higher transient and frequency-restoration costs than others.

\begin{figure}[ht]
\centering
\includegraphics[width=0.8\columnwidth]{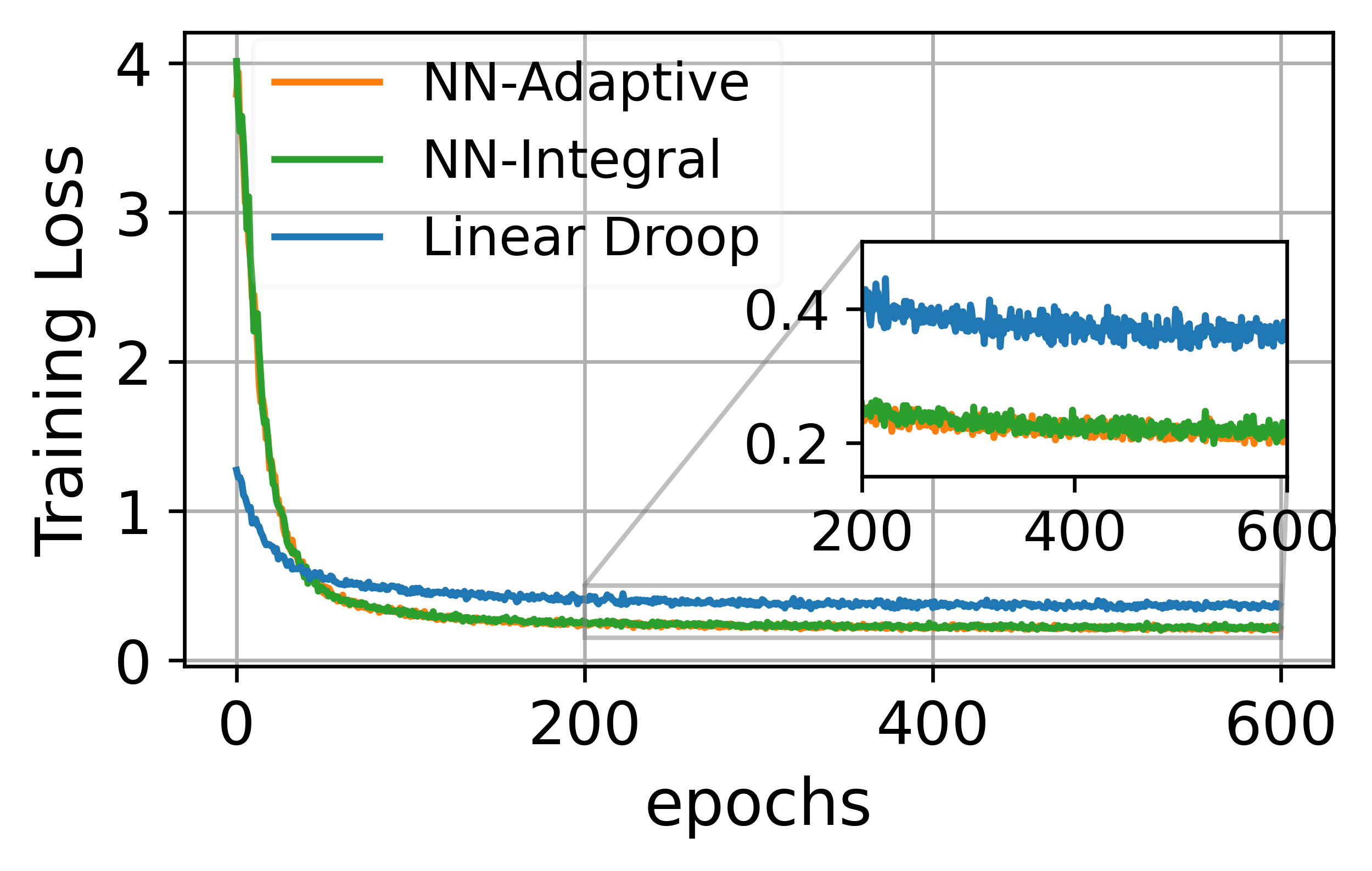}
\caption{ Average batch loss along epochs for IEEE-NE39 test case.  All converge,  with NN-Adaptive and NN-Integral achieving much lower loss than Linear Droop. The training loss only reflects transient cost but cannot reflect the performance of frequency restoration.
}
\label{fig:loss_n10}
\end{figure}

\begin{figure}[ht]
\centering
\includegraphics[width=0.8\columnwidth]{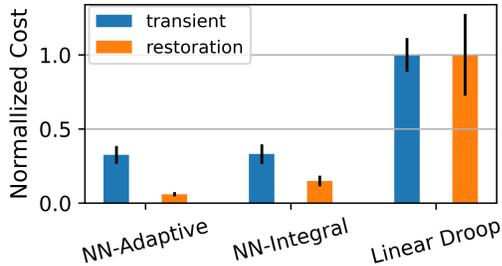}
\caption{ The average transient cost and frequency-restoration cost with error bar
on the randomly generated  test set with size 300. NN-Adaptive achieves  the lowest transient and frequency-restoration cost. 
}
\label{fig:cost_n10}
\end{figure}

It is noteworthy that the training loss (in Fig.~\ref{fig:loss_n10}) as well as the transient cost on the test set (in Fig.~\ref{fig:cost_n10})  for NN-Adaptive and NN-Integral is similar. But the frequency-restoration cost of NN-Adaptive on the test set is \emph{significantly lower} than that of NN-Integral.  The reason is that training only covers a finite length of trajectories where the transient cost is dominant, and it is difficult to quantify how long the trajectory should be to improve frequency-restoration performances. Therefore, the adaptive control approach provides an easy and efficient way to enforce frequency restoration by design, which also does not degrade transient performances. 

\begin{figure}[ht]
\centering
\subfloat[NN-Adaptive
]{\includegraphics[width=3.4
in]{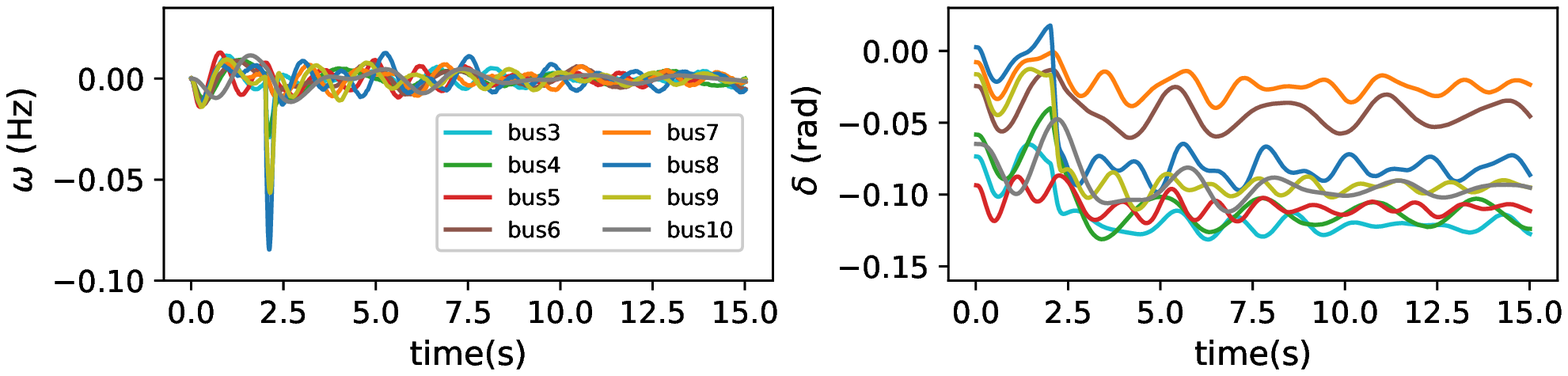}%
}
\hfil
\subfloat[NN-Integral
]{\includegraphics[width=3.4in]{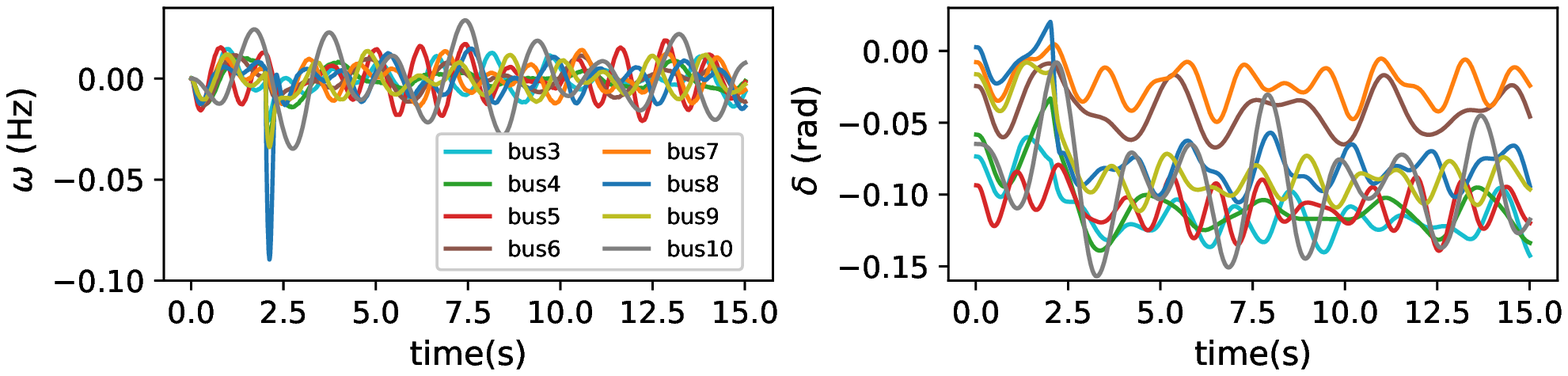}%
}
\hfil
\subfloat[Linear Droop
]{\includegraphics[width=3.4in]{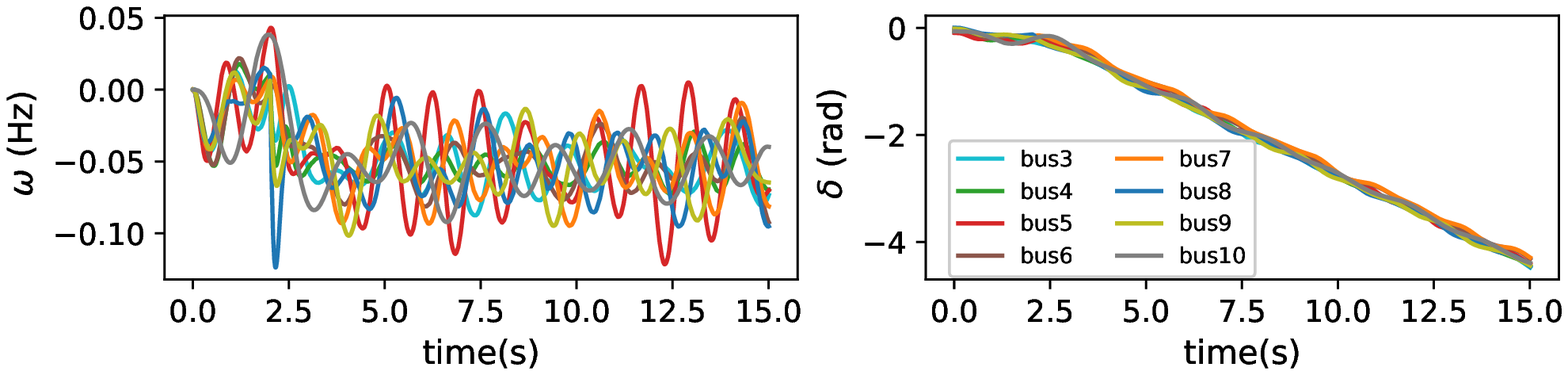}%
}
\caption{Dynamics of the frequency deviation $\omega$ and angle $\delta$ on 8 nodes with   a step load change at 2s. (a) NN-Adaptive achieves fast restoration of frequency after disturbances and the lowest oscillations. (b)  NN-Integral has higher oscillations. (c) Linear droop control cannot restore the frequency to its nominal value and has significantly larger oscillations. }
\label{fig:Dynamic_n10}
\end{figure}

To demonstrate the performance of different controllers, we show the dynamics of the frequency deviation $\bm{\omega}$ and angle $\bm{\delta}$ after  a step increase of load at the time of 2s in Fig.~\ref{fig:Dynamic_n10}. As guaranteed in Theorem~\ref{thm: convergence}, $\bm{\omega}$ of NN-Adaptive in Fig.~\ref{fig:Dynamic_n10}(a) converges to approximately zero after the step load change. NN-Adaptive in Fig.~\ref{fig:Dynamic_n10}(b) also pulls back the frequency deviation to the level close to zero, but it has much larger oscillations compared with NN-Adaptive. By contrast, linear droop control in Fig.~\ref{fig:Dynamic_n10}(c) cannot recover frequency at the nominal value and experiences significantly large  oscillations.   
Therefore, the adaptive control law can greatly improve both transient and frequency-restoration performances.

\subsection{IEEE 50-machine, 145-bus test case  }
To validate the performance of the proposed method in larger systems, we further conduct case studies on IEEE 145-bus, 50-machine
dynamic test case~\cite{chow1992toolbox}.  The parameter setting for the neural network controller and the training process is the same as the 39-bus system. Apart from the step change of load, here we add noise uniformly distributed in $\text{uniform}[-0.03,0.03]\,\text{p.u.}$ (where 1p.u.=100 MW) on the net power injections $p_i(t), i\in[n]$ 
to account for the time-varying deviations that are not covered in $\bm{\phi}_i(t)^\top \bm{a}_i$.

The average batch loss  during epochs of training is shown in Fig.~\ref{fig:loss_n50} and  the transient \& frequency-restoration cost on the test set is shown in
Fig.~\ref{fig:cost_n50}. Similar to the observations in IEEE 39-bus test system, NN-Adaptive has the lowest transient and frequency-restoration cost.  The transient cost of NN-Adaptive is 8\% lower than  NN-Integral and 69\% lower than Linear Droop, respectively.  The frequency-restoration cost of NN-Adaptive is 49\% lower than  NN-Integral and 88\% lower than Linear Droop, respectively. 
To visualize the performance improvement, we show the dynamics of the frequency deviation $\bm{\omega}$ and angle $\bm{\delta}$ after  a step increase of load at the time of 2s in Fig.~\ref{fig:Dynamic_n50}. Despite the existence of noise, $\bm{\omega}$ of NN-Adaptive in Fig.~\ref{fig:Dynamic_n50}(a) converges to approximately zero. NN-Integral in Fig.~\ref{fig:Dynamic_n50}(b) and linear droop control in Fig.~\ref{fig:Dynamic_n50}(c) experience much larger oscillations compared with NN-Adaptive. Moreover, droop control also cannot restore the frequency to its nominal value. Therefore, conventional droop control may not be enough to achieve satisfactory frequency response when there exists more frequent fluctuations of net load.

\begin{figure}[ht]
\centering
\includegraphics[width=0.8\columnwidth]{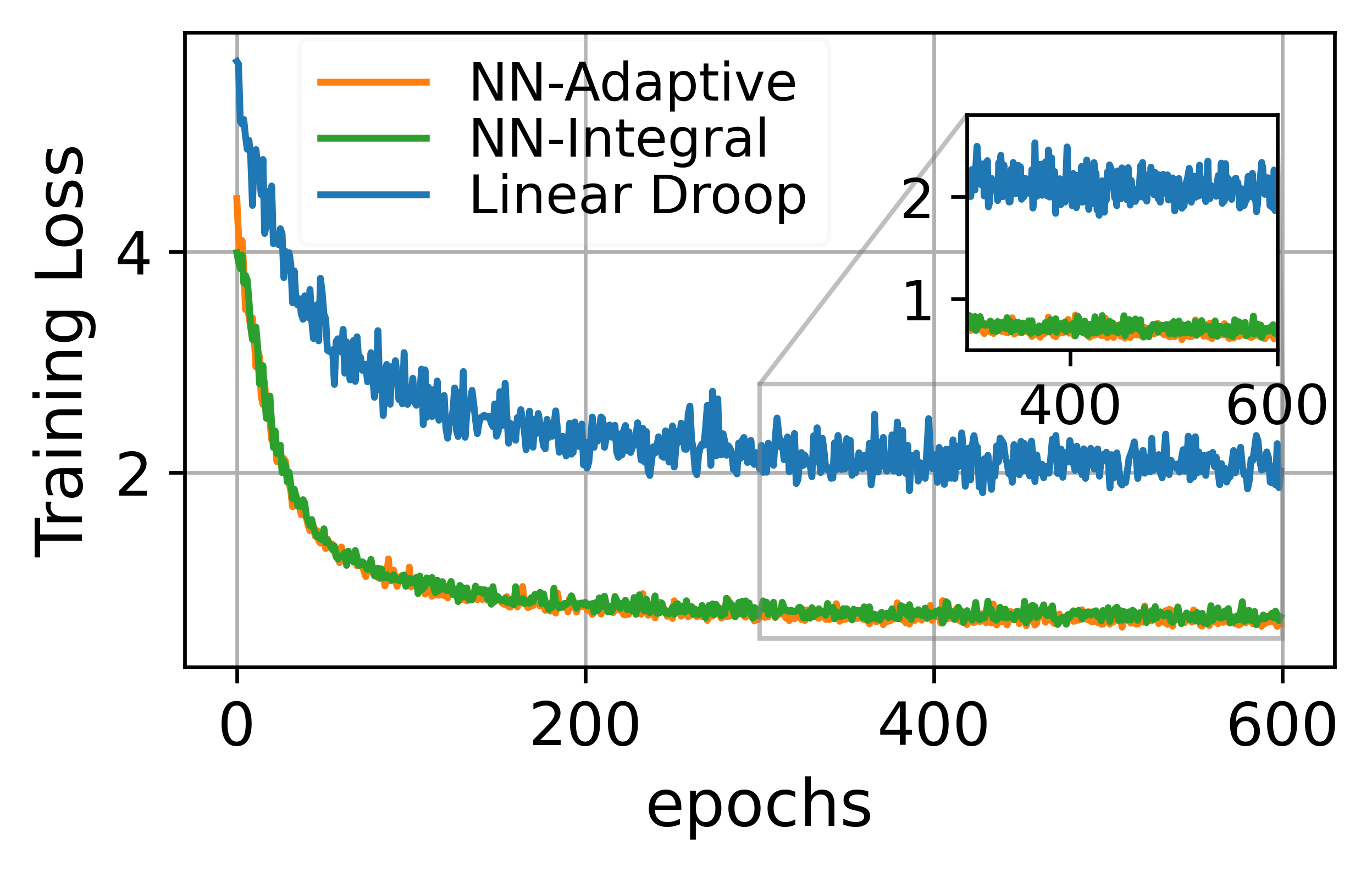}
\caption{ Average batch loss along epochs  for IEEE 145-bus test case. All converge,  with NN-Adaptive and NN-Integral achieving much lower loss than Linear Droop. The training loss only reflects transient cost but cannot reflect the performance of frequency restoration.
}
\label{fig:loss_n50}
\end{figure}

\begin{figure}[ht]
\centering
\includegraphics[width=0.8\columnwidth]{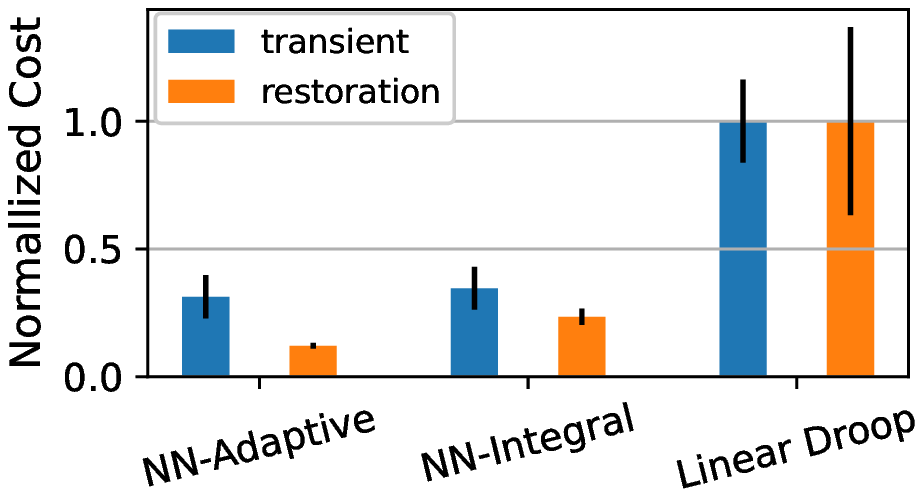}
\caption{ The average transient cost and frequency-restoration cost with error bar
on the randomly generated  test set with size 300. NN-Adaptive achieves  the lowest transient and frequency-restoration cost. 
}
\label{fig:cost_n50}
\end{figure}

\begin{figure}[ht]
\centering
\subfloat[NN-Adaptive
]{\includegraphics[width=3.4
in]{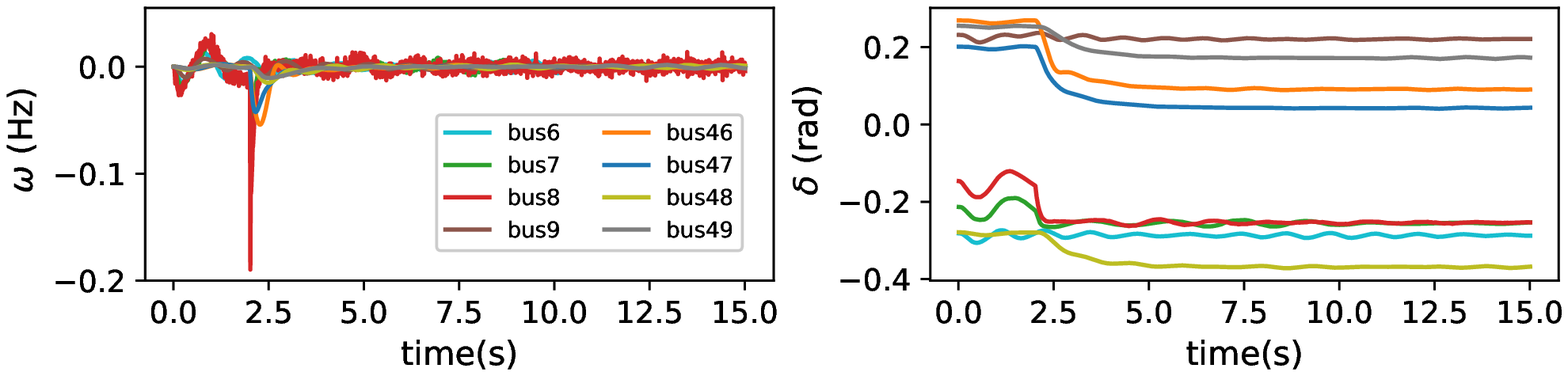}%
}
\hfil
\subfloat[NN-Integral
]{\includegraphics[width=3.4in]{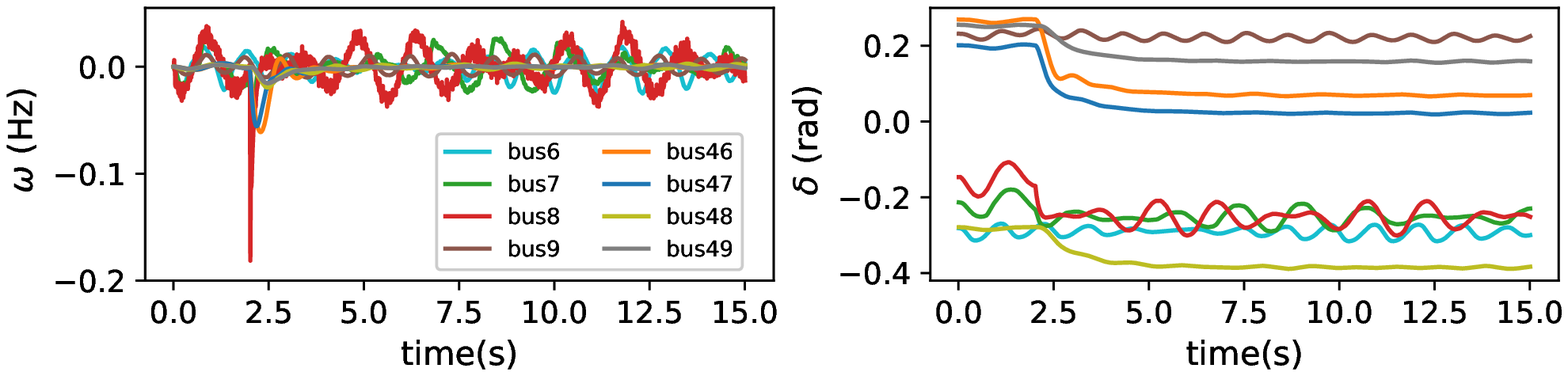}%
}
\hfil
\subfloat[Linear Droop
]{\includegraphics[width=3.4in]{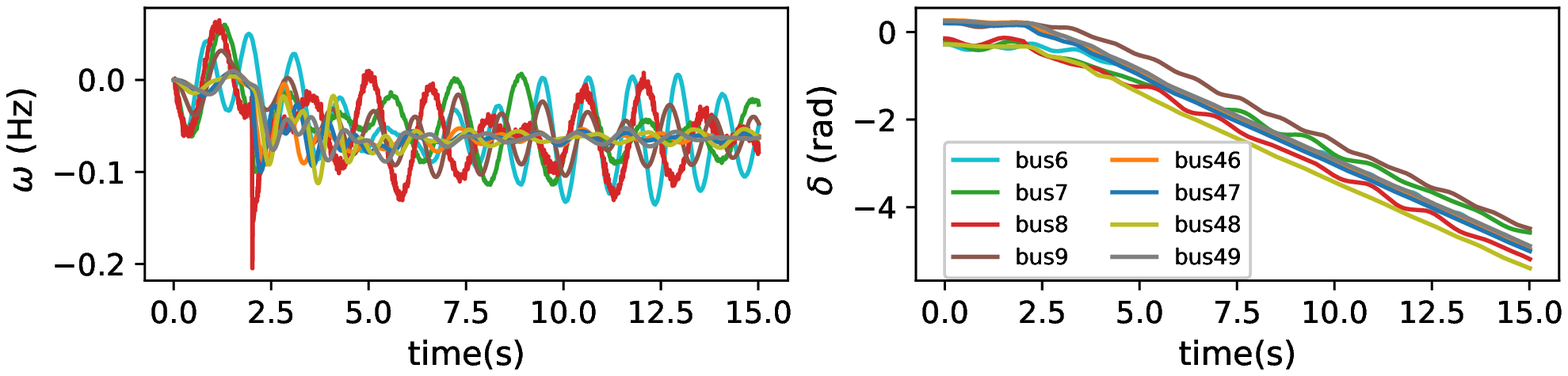}%
}
\caption{Dynamics of the frequency deviation $\omega$ and angle $\delta$ on 8 nodes with   a step load change at 2s. (a) NN-Adaptive achieves fast restoration of frequency after disturbances and the lowest oscillations. (b)  NN-Integral has higher oscillations. (c) Linear droop control cannot restore frequency to its nominal value and has large  oscillations. }
\label{fig:Dynamic_n50}
\end{figure}